\documentclass[aps,pra,onecolumn,nobibnotes,superscriptaddress,nofootinbib]{revtex4-2}
\usepackage{graphicx} 
\usepackage{caption,graphics,physics}
\usepackage{subcaption}
\usepackage{amsmath,xcolor,bm}
\usepackage{qcircuit}
\usepackage{mathtools,amsmath,amsthm,amsfonts,amssymb,amscd}
\usepackage{hyperref}
\usepackage[ruled,vlined]{algorithm2e}
\SetAlgorithmName{Algorithm}{Algorithm}{List of Protocols}
\usepackage{algpseudocode}

\usepackage[capitalize]{cleveref}
\crefname{algorithm}{Algorithm}{Protocols}
\hypersetup{
    colorlinks=true,
    linkcolor=red,
    filecolor=magenta,      
    urlcolor=cyan,
}
\usepackage{geometry}

\def\polylog{\mathrm{polylog}}

\newtheorem{theorem}{Theorem}

\newtheorem{lemma}{Lemma}

\newtheorem{remark}[theorem]{Remark}

\newtheorem{definition}{Definition}

\newtheorem{corollary}{Corollary}

\Crefname{appsec}{Appendix}{Appendices}

\begin{document}
\title{Exact and Efficient Circuit Construction for Block Encoding Matrix Polynomials}

\author{Taehee Ko}
\email{kthmomo@kias.re.kr}
\affiliation{School of Computational Sciences, Korea Institute for Advanced Study, Seoul, Republic of Korea}

\begin{abstract}
A recent interpolation-based Quantum Signal Processing (QSP) framework by Alase bypasses the phase-finding procedures required in conventional QSP, allowing for a direct encoding of the target polynomial into a quantum circuit. However, this approach assumes access to a diagonal block encoding of function values without providing an explicit circuit construction. In this work, we address this gap by developing an explicit quantum circuit construction method for diagonal block encodings. We prove that the resulting algorithm is numerically stable and constructs block-encoding circuits for matrix polynomials of degree \(d\) in \(\mathcal{O}(d\log d)\) classical time, improving on the best known theoretical bounds for this task. Numerical results confirm these theoretical findings. For polynomial degrees of order $10^7$, the proposed algorithm computes the phase factors in about a minute on a standard CPU. In double precision, the encoding error remains below $10^{-11}$ for all tested degrees up to $2^{16}$.
\end{abstract}

\maketitle

\section{Introduction}

Quantum Signal Processing (QSP) \cite{low2017optimal,gilyen2019quantum} provides an optimal framework for the block encoding of matrix polynomials of Hermitian matrices with minimal ancilla overhead. Despite its relatively short history, QSP has already become a central component of modern quantum algorithms across diverse areas of scientific computation, including quantum many-body problems and partial differential equations \cite{huang2026fullqubit,gilyen2019quantum,li2026quantum,lin2026quantum,martyn2021grand,rouze2026optimal,li2023efficient,chen2025efficient}. In addition, recent advances in quantum machine learning \cite{guo2024nonlinear,ivashkov2026qkan,khan2026quantum} highlight the potential of QSP-based techniques, which may also facilitate learning properties of quantum systems \cite{chen2022detecting,cho2024machine}. Because these applications rely on quantum circuits for QSP representations as a core primitive, explicit circuit constructions are essential for realizing and evaluating the practical advantages of quantum algorithms.

One difficulty of constructing explicit circuits for QSP lies in finding phase factors. To tackle this, several algorithms have been developed from diverse numerical perspectives, using root-finding procedures \cite{haah2019product,chao2020finding}, optimization subroutines \cite{ying2022stable,dong2024robust}, fixed-point iterations \cite{dong2024infinite,ni2026fast}, and the nonlinear Fourier transform \cite{alexis2026infinite,ni2026fast}. In particular, Ni and Ying \cite{ni2026fast} provided state-of-the-art algorithms for computing phase factors to a desired precision for real polynomials of a definite parity. This parity constraint is eliminated in generalized quantum signal processing (GQSP) \cite{motlagh2024generalized}, which guarantees the existence of a quantum circuit for an arbitrary Laurent polynomial. While several algorithms have been developed to construct GQSP quantum circuits \cite{yamamoto2024robust,berntson2025complementary,laneve2025generalized,ni2025inverse}, finding phase factors for arbitrary Laurent polynomials remains elusive. As a result, a method to explicitly construct quantum circuits for block encoding arbitrary Laurent polynomials of a matrix with single-qubit-controlled input unitary models is currently lacking.

In this work, we address this gap. We employ the recent QSP framework by Alase \cite{alase2025quantum}, which assumes oracle access to a diagonal block encoding of function values at interpolation points. To eliminate this assumption and explicitly construct the quantum circuit, we develop a method for constructing the diagonal block encoding. Our construction involves the multiplication of only two uniformly controlled rotations (UCRs), whose rotation angles can be computed exactly in $\mathcal{O}(d\log d)$ time. Accordingly, the resulting algorithm achieves a compilation time of $\mathcal{O}(d\log d)$, matching the best time observed in empirical results \cite{motlagh2024generalized} and improving upon the best-known theoretical bound in \cite{ni2026fast,ni2025inverse} by a factor of $\mathcal{O}(\log d)$. Furthermore, we prove that our algorithm is numerically stable, in the sense that it is robust against bit-precision errors in computing the required phase factors. We note that this approach requires $\mathcal{O}(\log d)$ ancilla qubits, a trade-off inherited from the underlying framework \cite{alase2025quantum}.

\section{Comparison to related work}

For arbitrary Laurent polynomials, a conventional LCU approach \cite{2481569.2481570} allows for the direct construction of an exact block encoding. However, due to the need for SELECT and PREPARE oracles \cite{lin2026quantum}, this approach requires multi-qubit controlled input unitary operations. In particular, the amplitude encoding in the PREPARE oracle requires the coefficients of a target Laurent polynomial to be normalized by their absolute sum. These shortcomings are addressed in the proposed method, which uses single-qubit controlled input unitary operations and imposes the relaxed normalization condition.  

GQSP \cite{motlagh2024generalized} is implemented using single-qubit controlled operations under the normalization condition with the uniform norm. One advantage of GQSP is that it uses a single ancilla qubit, whereas our approach requires $\mathcal{O}(\log d)$ ancilla qubits. However, constructing explicit GQSP circuits is not straightforward due to the task of finding phase factors, and rigorous guarantees of existing algorithms have limitations. For example, one should find an outer as a complementary polynomial for numerical stability \cite[Theorems 5.5 and 5.6]{ni2025inverse}, otherwise  standard algorithms can be numerically instable (e.g. see \cite[Section 5.2]{ni2025inverse}). Additionally, these stability results assume an exact pair of complementary polynomials, and it is not clear whether reasonable phase factors can be obtained for the case of approximate complementary polynomials \cite{berntson2025complementary,laneve2025generalized}.

For real polynomials of definite parity, rigorous guarantees have been established for various algorithms (e.g., see \cite[Table 1]{ni2026fast}). Using a single ancilla qubit, the two algorithms in \cite[Theorem 1 and Theorem 2]{ni2026fast} guarantee circuits to a target precision. In situations where a block-encoding input model is preferred for polynomials of definite parity (e.g., Gibbs state preparation with a shift \cite{gilyen2019quantum}, the eigenvalue thresholding problem \cite{martyn2021grand}), these algorithms may be more beneficial than ours in \cref{thm:chebyshevQSP} in the sense that they use uncontrolled block encodings. However, the opposite is true when a Laurent polynomial approximation without definite parity is required (e.g., extensions of Hamiltonian simulation \cite{motlagh2024generalized} and Liouvillian simulation \cite{huang2026fullqubit}). It is also worth noting that the normalization condition inherited from the framework employed \cite{alase2025quantum} is imposed only at the $4d$
interpolation points. This condition is weaker than the corresponding
normalization conditions used in LCU, GQSP, and
\cite[Theorem 1 and Theorem 2]{ni2026fast}. However, this relaxation comes at the cost of a subnormalization factor of $\sqrt{2}$ for the block encoding of a matrix polynomial by construction (\cref{lem: lemma 9}), thereby reducing the success probability by a factor of two. In other words, our quantum circuits can be regarded as ones in the non-fully-coherent regime of GQSP methods where $\norm{f_d}_\infty\le\frac{1}{\sqrt{2}}$. We summarize further details in \cref{tab:comparison}.

Our method for diagonal block encoding is comparable to the state-of-the-art single ancilla block encoding (SIABLE) \cite{li2026reducing}, assuming SIABLE is applied specifically to a diagonal matrix, as summarized in \cref{tab:diag_comparison}. SIABLE improves upon previous methods \cite{guseynov2025gate,li2025binary,rattew2023non,camps2022fable} across various complexity measures. Notably, it requires only a single ancilla qubit, unlike prior methods that require a system-sized ancilla register, and achieves a subnormalization factor of $1$. 

Because SIABLE targets a general matrix based on the SVD, its circuit consists
of a UCR encoding the singular values together with the singular-vector
unitaries. For a diagonal matrix, the latter reduce to a diagonal phase
unitary, which can be synthesized using $R_z$ and CNOT gates. In contrast,
our construction directly cascades two consecutive UCRs with $R_y$ and $R_z$
rotations. Due to this, gates for both methods are different, as shown in \cref{tab:diag_comparison}.
In contrast to these two methods, the recent framework \cite{xu2025unified} targets the low-level optimization of a class of uniformly controlled gates termed restricted uniformly controlled gates (rUCGs). Because both our method and SIABLE rely on two UCRs, which are a special case of rUCGs, these approaches are complementary: applying their circuit optimization technique \cite[Theorem 6]{xu2025unified} to the UCRs can further reduce the circuit depth from $\mathcal{O}(d)$ to $\mathcal{O}(\frac{d}{\log d})$.

\begin{table}[htbp]
    \centering
    \resizebox{\textwidth}{!}{%
    \begin{tabular}{|c|c|c|c|c|c|}
    \hline
    Algorithm & LCU \cite{2481569.2481570} &  \cite{ni2025inverse} & \cite[Theorem 1]{ni2026fast} & \cite[Theorem 2]{ni2026fast} & This work 
    (\cref{thm:complexQSP}) \\
    \hline
    Classical preprocessing time & $\mathcal{O}(d)$  & $\mathcal{O}(d\log^2 d)$\footnote{The stated time complexity assumes that an exact pair of complementary polynomials is given and the complementary polynomial satisfies an outerness condition.} & $\widetilde{\mathcal{O}}(d^2+\frac{d\log(\frac{d}{\eta\epsilon})}{\eta})$ & $\mathcal{O}(d\log^2d\log\frac{1}{\epsilon})$ & $\mathcal{O}(d\log d)$ \\
    \hline
    Algorithmic error & No & Yes\footnote{This refers to the error introduced by approximating the complementary polynomial when only the target polynomial is supplied (e.g. error bounds in \cite{berntson2025complementary}).} & Yes & Yes & No \\
    \hline
    Polynomial & Laurent  & Laurent & Real even (or odd) & Real even (or odd) & Laurent \\
    \hline
    Ancilla qubits & $\mathcal{O}(\log d)$ & $\mathcal{O}(1)$ & $\mathcal{O}(1)$ & $\mathcal{O}(1)$ & $\mathcal{O}(\log d)$ \\
    \hline
    Input model & MQC-$U_T$ & SQC-$U_T$  & $U_B$ & $U_B$ & SQC-$U_T$ \\
    \hline 
    Quantum query complexity & $\mathcal{O}(d)$ calls & $\mathcal{O}(d)$ calls & $\mathcal{O}(d)$ calls & $\mathcal{O}(d)$ calls & $\mathcal{O}(d)$ or $\mathcal{O}(\log d)$\footnote{This case applies to fast-forwardable Hamiltonian $H$ (e.g., efficiently-diagonalizable Hamiltonians \cite{li2025classical,gu2021fast})} calls \\
    \hline
    Normalization condition & $\norm{f_d}_1\le 1$\footnote{This $1$-norm is defined as the absolute sum of coefficients for a Laurent polynomial.} & $\norm{f_d}_\infty\le 1$ & $\norm{f_d}_\infty\le 1-\eta$ & $\norm{f_d}_{1}< 0.861$\footnote{This $1$-norm is defined as the absolute sum of Chebyshev coefficients for an even-degree polynomial (as defined in \cite[Theorem 2]{ni2026fast}).} & $\max_{0\le k\le 4d-1}\abs{f_d\left(e^{i2\pi k/4d}\right)}\le 1$ \\
    \hline 
    Subnormalization factor & $1$ & $1$ & $1$ & $1$ & $\sqrt{2}$ \\
    \hline
    Numerical stability & Yes & Yes\footnote{The results in \cite[Theorems 5.5 and 5.6]{ni2025inverse} assume an exact pair of complementary polynomials and an outerness condition.} & Yes & Yes & Yes \\
    \hline
    \end{tabular}%
    }
    \caption{Comparison of algorithms for block encoding matrix polynomials of degree $d$. ``Classical preprocessing time'' indicates the computational time of finding phase factors for quantum circuits or the compilation time of the coefficient-dependent PREPARE oracle for LCU. ``MQC-$U_T$'' and ``SQC-$U_T$'' stand for multi-qubit controlled and single-qubit controlled time-evolution circuits where $U_T=e^{iH}$. $U_B$ denotes a block encoding of $H$. ``Quantum query complexity'' counts the number of calls to the specified input model, excluding the cost of implementing the queried oracles. The parameter $\eta$  characterizes the near-coherent regime of QSP when $\eta$ is small.  The precision of phase factors, $\epsilon$, is defined respectively in \cite[Theorem 8]{alexis2026infinite} and \cite[Theorem 6]{dong2024infinite}. The methods in \cite{ni2026fast} construct symmetric-QSP circuits. We note that the parity constraint in \cite{ni2026fast} can in principle be removed by an LCU combination of the even and odd parts, at the cost of requiring controlled input-model circuits. }
    \label{tab:comparison}
\end{table}

\begin{table}[h]
\centering
\begin{tabular}{|c|c|c|}
\hline
 Method & This work (\cref{cor:complex_diag_encoding})  & \cite[SIABLE]{li2026reducing}\\
\hline
CNOT & $8d-2$ & $8d-2$ \\
 \hline
 Single-qubit gates & $8d$ & $8d+1$ \\
 \hline
R$_y$ gate & $4d$ & $0$ \\
\hline
R$_z$ gate & $4d$ & $8d-1$ \\
\hline
Hadamard gate & $0$ & $2$  \\
\hline
Total gate & $16d-2$ & $16d-1$ \\
\hline
Ancilla qubit & $1$  & $1$ \\
\hline
Circuit depth & $\mathcal{O}(d)$ & $\mathcal{O}(d)$ \\
\hline
Subnormalization factor & $1$ & $1$ \\
\hline
\end{tabular}
\caption{
Comparison between our approach and SIABLE
specialized to a $2^{m+2}$-dimensional diagonal matrix, where
$d=2^m$. The SIABLE gate counts are obtained by specializing its circuit construction
to diagonal matrices, for which the singular-vector unitary reduces to a
diagonal unitary. }
\label{tab:diag_comparison}
\end{table}

\section{Main Result}
\label{sec:results}

For a Hermitian matrix $H$, we consider two well-known input models: a block encoding $U_B = \begin{pmatrix} H & * \\ * & * \end{pmatrix}$ with $\norm{H}_2\le 1$, and a Hamiltonian time-evolution circuit $U_T = e^{iH}$ with $\norm{H}_2\le \pi$. For each model, our goal is to construct an explicit circuit for a block encoding of a matrix polynomial of $H$. We first recall the standard definition of a block encoding.

\begin{definition}[Block encoding \cite{gilyen2019quantum}]\label{def: block encoding}
For a $2^n \times 2^n$ matrix $A$ and $\alpha_A \in \mathbb{R}^+$, we say that an $(n+a_A)$-qubit unitary operator $U_A$ is an $(\alpha_A,a_A)$-block encoding of $A$ if 
\begin{equation}
    A = \alpha_A(\mathbb{I}_n\otimes\bra{0^{a_A}})U_A(\mathbb{I}_n\otimes\ket{0^{a_A}}),
\end{equation}
where $\ket{0^{a_A}}$ is the all-zero state of the $a_A$-qubit ancilla register. $\alpha_A$ is  called the subnormalization factor.
\end{definition}

We summarize our main results as follows. For both results, the core block-encoding circuit is specified in \eqref{eq:complexUBE}. We will assume that the degree of polynomial, $d$, is a power of $2$ for convenience; otherwise, $d$ can be replaced
with its ceiling. The following theorem is based on \cite[Theorem 4]{alase2025quantum}. The proofs are deferred to \cref{sec:qsp}. 

\begin{theorem}[QSP with Hamiltonian time-evolution circuit $U_T$]
\label{thm:complexQSP}
Let $f_d:S^1\to\mathbb{C}$ be a Laurent polynomial of degree $d=2^m$ satisfying
\begin{equation}
\label{eq:complexfnormalization}
    \max_{0\le k\le 4d-1}\abs{f_d\left(e^{i2\pi k/4d}\right)}\le 1.
\end{equation}
Then, given any $n$-qubit unitary $U_T=e^{iH}$, one can construct an $(\sqrt{2},m+3)$-block encoding of $f_d(e^{iH})$ using $4d-1$ queries each to controlled-$U_T$ and controlled-$U_T^\dagger$, along with $\mathcal{O}(d)$ additional single- and two-qubit gates, in $\mathcal{O}(d\log d)$ classical time.
\end{theorem}

The following corollary specializes \cite[Theorem 5]{alase2025quantum}
to block encodings of matrix polynomials and specifies the classical construction time
using our explicit diagonal block encoding in \cref{sec:diag}.
\begin{corollary}[QSP with block encoding $U_B$]
\label{thm:chebyshevQSP}
Let $p_d(x):[-1,1]\to\mathbb{C}$ be a polynomial of degree $d=2^m$ satisfying 
\begin{equation}
\label{eq:eta_cheb}
    \max_{0\le r\le 4d-1}\abs{p_d\left(\cos\frac{2\pi r}{4d}\right)}\le 1.
\end{equation}
Then, given a $(1,a_H)$-block encoding of $H$, $U_B$, one can construct an $(\sqrt{2},m+4+a_H)$-block encoding of $p_d(H)$ using at most $4d-1$ queries each to controlled-$U_B$ and controlled-$U_B^\dagger$, along with $\mathcal{O}(d a_H)$ additional single- and two-qubit gates, in $\mathcal{O}(d\log d)$ classical time.
\end{corollary}

The normalization conditions in \eqref{eq:complexfnormalization} and \eqref{eq:eta_cheb} are induced from our circuit constructions for diagonal block encodings. We discuss this further in the following section.

\subsection{Diagonal encoding from unitary controlled rotations}\label{sec:diag}

We introduce a  method for constructing a diagonal block encoding for an arbitrary complex diagonal matrix using unitary controlled rotations (UCRs). We then apply this general construction to our application, that is, encoding the function values of a Laurent polynomial.

Consider a general $N \times N$ diagonal matrix $D = \sum_{k=0}^{N-1} \lambda_k \ketbra{k}_J$, satisfying $\max_k |\lambda_k| \le 1$. Let us write each diagonal entry in polar form as $\lambda_k = |\lambda_k|e^{i\theta_k}$. To proceed, we define the rotation angles
\begin{equation}\label{eq:general_angles}
    r_k := |\lambda_k| \in [0,1], \qquad
    \beta_k := 2\arccos r_k, \qquad
    \gamma_k := -2\theta_k.
\end{equation}
Using these angles, we construct the following UCRs with rotation $Y$'s ($R_y$) and rotation $Z$'s ($R_z$),
\begin{align}
    F^{(y)}_{J,a}(\boldsymbol{\beta})
    &:=
    \sum_{k=0}^{N-1}
    \ketbra{k}_J
    \otimes
    R_y(\beta_k)_a,
    \\
    F^{(z)}_{J,a}(\boldsymbol{\gamma})
    &:=
    \sum_{k=0}^{N-1}
    \ketbra{k}_J
    \otimes
    R_z(\gamma_k)_a,
\end{align}
where $\boldsymbol{\beta}:=(\beta_0,\ldots,\beta_{N-1})$, $\boldsymbol{\gamma}:=(\gamma_0,\ldots,\gamma_{N-1})$, and $a$ denotes a single ancilla qubit.

\begin{lemma}[General complex diagonal encoding]
\label{lem:general_diagonal_encoding}
For any diagonal matrix $D = \sum_{k=0}^{N-1} \lambda_k \ketbra{k}_J$ with $\max_k |\lambda_k| \le 1$, the unitary
\begin{equation}
\label{eq:general_direct_complex_diag}
    U_D
    :=
    F^{(z)}_{J,a}(\boldsymbol{\gamma})
    F^{(y)}_{J,a}(\boldsymbol{\beta})
\end{equation}
is a $(1,1)$-block encoding of $D$, namely,
\begin{equation}
\label{eq:general_direct_complex_diag_block}
    \left(
        \mathbb{I}_J\otimes\bra{0}_a
    \right)
    U_D
    \left(
        \mathbb{I}_J\otimes\ket{0}_a
    \right)
    =
    D.
\end{equation}
\end{lemma}

\begin{proof}
By construction, 
\begin{equation}
    U_D
    =
    \sum_{k=0}^{N-1}
    \ketbra{k}_J
    \otimes
    R_z(\gamma_k)R_y(\beta_k).
\end{equation}
Using the conventional definitions of rotation $Y$ and $Z$ gates,
\begin{equation}
    R_y(\beta)
    =
    \begin{pmatrix}
        \cos(\beta/2) & -\sin(\beta/2)\\
        \sin(\beta/2) & \cos(\beta/2)
    \end{pmatrix},
    \qquad
    R_z(\gamma)
    =
    \begin{pmatrix}
        e^{-i\gamma/2} & 0\\
        0 & e^{i\gamma/2}
    \end{pmatrix},
\end{equation}
we have
\begin{equation}
    \bra{0}
    R_z(\gamma_k)R_y(\beta_k)
    \ket{0}=
    e^{-i\gamma_k/2}
    \cos\left(\frac{\beta_k}{2}\right)
    =
    e^{i\theta_k}r_k
    =
    \lambda_k,
\end{equation}
by the definition in \eqref{eq:general_angles}. Finally, we observe
\begin{equation}
\left(
    \mathbb{I}_J\otimes\bra{0}_a
\right)
U_D
\left(
    \mathbb{I}_J\otimes\ket{0}_a
\right)=
\sum_{k=0}^{N-1}
\ketbra{k}_J\,
\bra{0}
R_z(\gamma_k)R_y(\beta_k)
\ket{0}=
\sum_{k=0}^{N-1}
\lambda_k\ketbra{k}_J = D,
\end{equation}
which completes the proof.
\end{proof}
We now apply \cref{lem:general_diagonal_encoding} to our specific target application as in the following corollary.

\begin{corollary}[Complex diagonal encoding for Laurent polynomials]
\label{cor:complex_diag_encoding}
Let $f_d \not\equiv 0$ be a Laurent polynomial, and define the normalization factor $\eta := \max_{0\le k\le 4d-1}|f_d(z_k)|$ where $z_k=e^{i2\pi k/4d}$. By defining the rotation angle vectors $\boldsymbol{\beta}$ and $\boldsymbol{\gamma}$ from the normalized entries $\lambda_k := \frac{f_d(z_k)}{\eta}$ according to \eqref{eq:general_angles}, the unitary
\begin{equation}
\label{eq:U_f_4d}
    U_{f,4d}^{(\mathbb C)} := F^{(z)}_{J,a}(\boldsymbol{\gamma}) F^{(y)}_{J,a}(\boldsymbol{\beta})
\end{equation}
yields an $(\eta, 1)$-block encoding of the target diagonal matrix $\sum_{k=0}^{4d-1} f_d(z_k)\ketbra{k}_J$. Namely,
\begin{equation}
\label{eq:direct_complex_diag_block_app}
    \left(
        \mathbb{I}_J\otimes\bra{0}_a
    \right)
    U_{f,4d}^{(\mathbb C)}
    \left(
        \mathbb{I}_J\otimes\ket{0}_a
    \right)
    =
    \frac{1}{\eta}
    \sum_{k=0}^{4d-1}
    f_d(z_k)\ketbra{k}_J.
\end{equation}
Furthermore, the block encoding $U_{f,4d}^{(\mathbb C)}$ can be constructed using $8d$ single-qubit rotations and $8d-2$ CNOT gates, with the required rotation angles computed in $\mathcal{O}(d\log d)$ time.
\end{corollary}

\begin{proof}
Because $\lambda_k = \frac{f_d(z_k)}{\eta}$, we are guaranteed that $|\lambda_k| \le 1$ for all $k \in [4d]$. Applying \cref{lem:general_diagonal_encoding} for a system of size $N=4d$ immediately establishes that $U_{f,4d}^{(\mathbb C)}$ is a $(1,1)$-block encoding of $\sum_{k=0}^{4d-1} \lambda_k \ketbra{k}_J$, which is equivalent to \eqref{eq:direct_complex_diag_block_app}. The unitary operators $F^{(y)}_{J,a}(\boldsymbol{\beta})$ and $F^{(z)}_{J,a}(\boldsymbol{\gamma})$ 
can each be implemented using  $4d$ single-qubit rotations and $4d$ CNOTs \cite{mottonen2004quantum} before the boundary cancellation. As detailed in \cref{app:ucr_details}, evaluating the required rotation angles from the function values takes $\mathcal{O}(d\log d)$ time, completing the proof.
\end{proof}

\begin{remark}
    We note that the diagonal encoding, $U_D$, in \eqref{eq:general_direct_complex_diag} is further simplified in the case of real diagonal where $\gamma_k=0$ or $\gamma_k=\pi$. Since the rotation $Z$'s becomes the identity up to a global phase factor, $U_D=F_{J,a}^{(y)}(\beta)$. As a result, the block encoding of function values \eqref{eq:U_f_4d} is simplified to the one only with UCR with rotation $Y$'s in the case of real Laurent polynomial.
\end{remark}

\subsection{The proof of the main result}
\label{sec:qsp}
In this section, we prove \cref{thm:complexQSP} and \cref{thm:chebyshevQSP}.

\begin{proof}[Proof of \cref{thm:complexQSP}]

Under the normalization condition
\eqref{eq:complexfnormalization}, we apply
\cref{lem:general_diagonal_encoding}. Together with \cref{lem: lemma 9}, we obtain the desired block encoding as follows,
\begin{equation}
\label{eq:complexUBE}
    U_{BE}^{(\mathbb C)}
    :=
    \left(
    \mathbb{I}_S\otimes P_{\mathrm{out}}^\dagger\otimes \mathbb{I}_{a}
    \right)
    \left(
    \mathcal W_U^{(S,J)}\otimes \mathbb{I}_{a}
    \right)^\dagger
    \left(
    \mathbb{I}_S\otimes U_{\mathrm{diag}}^{(\mathbb C)}
    \right)
    \left(
    \mathcal W_U^{(S,J)}\otimes \mathbb{I}_{a}
    \right)
    \left(
    \mathbb{I}_S\otimes P_{\mathrm{in}}\otimes \mathbb{I}_{a}
    \right),
\end{equation}where
\begin{align}
    P_{\mathrm{in}} 
    & = \mathrm{H}^{\otimes (m+2)}, \\
    P_{\mathrm{out}} 
    &= \text{C}_0\text{NOT}_{(m+1) \to m} (\mathrm{H} \otimes \mathbb{I} \otimes \mathrm{H}^{\otimes m}), \\
    \mathcal{W}_U^{(S,J)}
    &=
    \prod_{r=0}^{m+1}
    \left(
        \mathbb{I}_S\otimes\ket{0}\!\bra{0}_{J_r}
        +
        U^{2^r}\otimes\ket{1}\!\bra{1}_{J_r}
    \right) = \sum_{j=0}^{4d-1} U^j\otimes\ket{j}\!\bra{j}_J,\\
    U_{\mathrm{diag}}^{(\mathbb C)}
    &=
    (\operatorname{QFT}_J \otimes \mathbb{I}_a)
    U_{f_d,4d}^{(\mathbb C)}
    (\operatorname{QFT}_J^\dagger\otimes\mathbb{I}_a).
\end{align} 
Here, $\mathrm{H}$ denotes the Hadamard gate. The unitary operators $P_{\mathrm{in}}$ and $P_{\mathrm{out}}$ acting on $J$ are defined such that
$P_{\mathrm{in}}\ket{0^{m+2}}_J = \ket{+_{4d}}_J$ and $P_{\mathrm{out}}\ket{0^{m+2}}_J = \ket{+_{2d}}_J$ according to \eqref{eq:plus-states}. Additionally, $\text{C}_0\text{NOT}_{(m+1) \to m}$ denotes a zero-controlled NOT gate controlled by the $(m+1)$-th qubit and acting on the $m$-th qubit as the target. $U_T$ is a given unitary acting on the system register $S$ in \cref{thm:complexQSP}. 

We can construct the diagonal block encoding $U_{f_d,4d}^{(\mathbb C)}$ in $\mathcal{O}(d\log d)$ time by \cref{lem:general_diagonal_encoding}. If a target Laurent polynomial does not initially satisfy the  normalization condition \eqref{eq:complexfnormalization}, we can normalize it after evaluating the required function values $\{f_d(z_k)\}_{k=0}^{4d-1}$, and construct the corresponding diagonal block encoding as in \cref{cor:complex_diag_encoding}. Evaluating those function values can be done using the fast Fourier transform (FFT) in $\mathcal{O}(d\log d)$ time.  This completes the proof.
\end{proof}

Using the above results, we now prove \cref{thm:chebyshevQSP}. Assume $\|H\|\le 1$, and let $U_B$
be a $(1,a_H)$-block encoding of $H$; that is,
\begin{equation}
    H
    =
    \left(
        \bra{0^{a_H}}\otimes \mathbb{I}
    \right)
    U_B
    \left(
        \ket{0^{a_H}}\otimes \mathbb{I}
    \right).
\end{equation}
According to \cite[Theorem 2]{low2019hamiltonian}, one can construct a qubitization walk operator $W_H$ using one query each to $U_B$ and $U_B^\dagger$, at most one additional ancilla qubit, and $\mathcal{O}(a_H)$ single- and two-qubit gates, such that for any integer $k \ge 0$,
\begin{equation}
\label{eq:walk_chebyshev_sandwich}
    \left(
        \bra{0^{a_H+1}}\otimes \mathbb{I}
    \right)
    W_H^k
    \left(
        \ket{0^{a_H+1}}\otimes \mathbb{I}
    \right)
    =
    T_k(H),
\end{equation}
where $T_k$ is the degree-$k$ Chebyshev polynomial of the first kind\footnote{As observed in \cite[Eq (7)]{low2019hamiltonian}, $W_H^k$ can be represented in terms of rotation $Y$'s with a phase factor of $k\arccos(\lambda)$ for each eigenvalue $\lambda$ of $H$. According to their notation, this allows us to rewrite $W_H^k=\begin{pmatrix}
    \cos (k\arccos(\lambda)) & *\\
    * & *\\
\end{pmatrix}_\lambda$, which corresponds to \eqref{eq:walk_chebyshev_sandwich}. }.   

We consider a degree-$d$ polynomial in terms of its Chebyshev expansion and its corresponding Laurent polynomial
\begin{equation}
\label{eq:cheb_poly}
    p_d(x)
    =
    \sum_{k=0}^{d} a_k T_k(x),\quad f_d(z)
    :=
    a_0
    +
    \frac{1}{2}
    \sum_{k=1}^{d}
    a_k
    \left(
        z^k+z^{-k}
    \right).
\end{equation}
Since $T_k\left(\frac{z+z^{-1}}{2}\right)
    =
    \frac{z^k+z^{-k}}{2}$ for any $z\in\mathbb{C}$, we have $f_d(z)
    =
    p_d\left(
        \frac{z+z^{-1}}{2}
    \right)$.
Applying \eqref{eq:walk_chebyshev_sandwich}, we observe
\begin{align}
\left(
    \bra{0^{a_H+1}}\otimes\mathbb{I}
\right)
f_d(W_H)
\left(
    \ket{0^{a_H+1}}\otimes\mathbb{I}
\right) &=
a_0\mathbb{I}
+
\frac{1}{2}
\sum_{k=1}^{d}
a_k
\left[
    T_k(H)+T_k(H)
\right] \nonumber \\
&= p_d(H),
\label{eq:cheb_signal_block}
\end{align}since 
$W_H^{-k}=(W_H^k)^\dagger$ and
$T_k(H)^\dagger=T_k(H)$.

Now we are ready to prove \cref{thm:chebyshevQSP}.

\begin{proof}[Proof of \cref{thm:chebyshevQSP}]

Since $f_d(z)
    =
    p_d\left(
        \frac{z+z^{-1}}{2}
    \right)$, the condition \eqref{eq:eta_cheb} naturally implies \eqref{eq:complexfnormalization}.
As shown in the proof of \cref{thm:complexQSP}, one can find an explicit
$(\sqrt{2},m+3)$-block encoding
$U_{BE}^{(\mathbb C)}(W_H)$ of $f_d(W_H)$ in $\mathcal{O}(d\log d)$ time. Moreover, using \eqref{eq:cheb_signal_block}, we notice
\begin{align}
&
\left(
    \bra{0^{m+3}}
    \otimes
    \bra{0^{a_H+1}}
    \otimes\mathbb{I}
\right)
U_{BE}^{(\mathbb C)}(W_H)
\left(
    \ket{0^{m+3}}
    \otimes
    \ket{0^{a_H+1}}
    \otimes\mathbb{I}
\right)
\nonumber\\
&\qquad=
\frac{1}{\sqrt{2}}
\left(
    \bra{0^{a_H+1}}\otimes\mathbb{I}
\right)
f_d(W_H)
\left(
    \ket{0^{a_H+1}}\otimes\mathbb{I}
\right)
\nonumber\\
&\qquad=
\frac{p_d(H)}{\sqrt{2}}.
\end{align}
Since the circuit structure in this case is identical to that in \cref{thm:complexQSP}, the total circuit complexity depends on that of $W_H$. By \cite[Theorem 2]{low2019hamiltonian}, we conclude that the
$(\sqrt{2},m+4+a_H)$-block encoding of $p_d(H)$ uses at most $4d-1$ queries each to controlled-$U_B$ and controlled-$U_B^\dagger$, along with $\mathcal{O}(d a_H)$ additional single- and two-qubit gates.

If a target polynomial $p_d(x)$ does not initially satisfy the normalization condition \eqref{eq:eta_cheb}, we can normalize it by evaluating the function values $\left\{p_d\left(\cos\frac{2\pi r}{4d}\right)\right\}_{r=0}^{4d-1}$. This requires only $\mathcal{O}(d\log d)$ time using the discrete cosine transform (DCT), because the $4d$ interpolation points correspond to Chebyshev nodes. This completes the proof.
\end{proof}

\section{Discussion}

We have proposed a circuit compilation algorithm for constructing block encodings of matrix polynomials, as summarized in \cref{alg:qsp}.

 \begin{algorithm}[h]
\SetAlgoLined
    \KwData{A target polynomial of degree $d=2^m$ (a Laurent polynomial $f_d(z)$ for Hamiltonian time-evolution circuit $U_T=e^{iH}$, or a complex polynomial $p_d(x)$ for block encoding $U_B$).}
    \KwResult{An explicit circuit description for $U_{BE}^{(\mathbb C)}$, an exact block encoding of the normalized target polynomial.}

    1. \textbf{Normalization \& Input Conversion:} Evaluate the polynomial at $4d$ interpolation points using FFT or DCT, and scale it to satisfy the normalization condition in \eqref{eq:complexfnormalization} or \eqref{eq:eta_cheb}. If the input is $p_d(x)$, convert it via $f_d(z) = p_d\left(\frac{z+z^{-1}}{2}\right)$ and set the base unitary $U = W_H$ \eqref{eq:walk_chebyshev_sandwich}. Otherwise, set $U = e^{iH}$.

    2. \textbf{Diagonal Block Encoding (\cref{cor:complex_diag_encoding}):} From the normalized function values $f_d(z_k) = |f_d(z_k)|e^{i\theta_k}$, construct $U_{f,4d}^{(\mathbb C)} = F^{(z)}_{J,a}(\boldsymbol{\gamma}) F^{(y)}_{J,a}(\boldsymbol{\beta})$ on the index register $J$ and ancilla $a$. Form the diagonal encoding $U_{\mathrm{diag}}^{(\mathbb C)} = (\operatorname{QFT}_J\otimes \mathbb{I}_a) U_{f,4d}^{(\mathbb C)} (\operatorname{QFT}_J^\dagger\otimes \mathbb{I}_a)$.

    3. \textbf{Final Circuit Construction:} Assemble the full block encoding circuit $U_{BE}^{(\mathbb C)}$ as defined in \eqref{eq:complexUBE} using the base unitary $U$.
     
    \caption{Exact QSP Circuit Construction}
    \label{alg:qsp}
\end{algorithm}

\subsection{Numerical Stability of \cref{alg:qsp}}

Steps 1 and 2 of \cref{alg:qsp} introduce floating-point errors due to the operations including the FFT, FWHT,
and inverse trigonometric functions. Despite this, we prove that our algorithm is still numerically stable. To this end, we adapt the notion of numerical stability introduced in
\cite[Section 5.1]{ni2025inverse} to our setting by defining
a stability criterion for constructing block encodings of
matrix polynomials.

\begin{definition}\label{def: stable}
An algorithm is said to be numerically stable if,
for any target accuracy $\epsilon>0$, it outputs
a quantum circuit corresponding to an approximation
$\widetilde f(z)$ of the target Laurent polynomial
$f_d(z)$ using
$r=\mathcal{O}\left(\polylog(d,\frac{1}{\epsilon})\right)$
bits of precision, such that
\begin{align}
    \frac{
    \max_{0\le k\le 4d-1}
    \abs{\widetilde f(z_k)-f_d(z_k)}
    }{
    \max_{0\le k\le 4d-1}
    \abs{f_d(z_k)}
    }
    \le \epsilon,
\end{align}
where the interpolation points $z_k$ are defined in
\cref{cor:complex_diag_encoding}.
\end{definition}

We note that, according to
\cite[Lemma 6]{alase2025quantum},
the number of points on the unit circle can be
reduced to $2d+1$. Nevertheless, we retain the
$4d$ interpolation points used in our circuit
construction for consistency, without affecting
the following analytic result.

We assume that all quantum gates are
implemented exactly, as in the  analysis
of \cite{ni2025inverse}. Thus, we consider only
floating-point errors arising from the computation
of rotation angles during classical preprocessing
(Steps 1 and 2 of \cref{alg:qsp}).
\begin{theorem}\label{thm:numerical_stability}
Assume that all quantum gates are implemented exactly. Then, for any target accuracy $\epsilon>0$,
\cref{alg:qsp} is numerically stable using
$\mathcal{O}\left(\log d+\log\frac{1}{\epsilon}\right)$
bits of floating-point precision.
\end{theorem}
The proof is provided in \cref{app:proof thm3}.
Here we briefly overview its technical detail.
Our analysis tracks the propagation of
floating-point errors through each step of 
numerical operations.
In Step 1 of \cref{alg:qsp}, the function values
at the interpolation points are evaluated using
the FFT, whose floating-point errors can be
bounded using \cite[Theorem 24.2]{higham2002accuracy}.
These approximate function values are then used
to compute the rotation angles in
\eqref{eq:general_angles} through elementary
functions. One might be concerned that this computation is numerically unstable, as the inverse trigonometric and phase functions can be sensitive near the boundaries of their domains. However, the computed angles are subsequently transformed linearly via the FWHT and then fed into Pauli rotations, which are Lipschitz continuous with respect to their rotation angles. Thus, these properties allow us to control the propagation of numerical errors. As shown in \eqref{ineq: u-u}, the resulting diagonal block encoding remains a good approximation of the target. Since all quantum gates are assumed
to be implemented exactly, this error bound is retained as in \eqref{ineq: last} and our algorithm is proven to be numerically stable.

\subsection{Numerical test}

For demonstration, we conducted numerical simulations on Google Colaboratory using an Intel(R) Xeon(R) CPU @ 2.20GHz (2 vCPUs) with 12 GB of system RAM (CPU-only runtime). The implementation used Python 3.13.18, NumPy 2.4.6, and SciPy 1.16.3. We used Joblib 1.6.0 to run independent trials in parallel.  All computations used IEEE 754 double precision (complex128 for complex values and float64 for real values), except in simulations where we explicitly varied the numerical precision. For large polynomial degrees, the peak memory usage per worker process exhibits linear $\mathcal{O}(d)$ scaling, reaching approximately 0.61 GB  at $d = 2^{21}$, 1.17 GB at $d = 2^{22}$, 2.29 GB at $d = 2^{23}$, and 4.54 GB at the largest tested degree $d = 2^{24}$. For each polynomial degree $d=2^m$, with $m=1,\ldots,24$ (corresponding to $d$ up to $2^{24} \approx 1.67\times 10^7$), we generated $10$ independent random complex vectors to simulate the construction of the diagonal block encoding $U_{f,4d}^{(\mathbb C)}$. Specifically, the required parameters are the UCR angle vectors $\boldsymbol{\beta}$ and $\boldsymbol{\gamma}$ extracted from $f_d(z_k)$. To simulate this, we drew unnormalized vectors $\widetilde{\boldsymbol{f}} \in \mathbb{C}^{4d}$ with entries $\widetilde f_k=x_k+i y_k$, where $x_k,y_k\overset{\mathrm{i.i.d.}}{\sim}\mathcal{N}(0,1)$. We then normalized these vectors using the standard $\ell_\infty$-norm, setting $\boldsymbol{f} = \frac{\widetilde{\boldsymbol{f}}}{\norm{\widetilde{\boldsymbol{f}}}_\infty}$. In the simulation, the $k$-th entry of the vector $\boldsymbol{f}$
corresponds to the value $f_d(z_k)$, where $z_k$ is the root of unity
defined in \cref{cor:complex_diag_encoding}. 

Passing these vectors into the second step of \cref{alg:qsp}, we recorded the execution times required to compute the explicit circuit parameters (i.e., the single-qubit rotation angles $\boldsymbol{\vartheta}^{(\nu)}$ in \eqref{eq:thetas}), as presented in \cref{fig:scaling}. We observe that the empirical scaling asymptotically converges to the theoretical $\mathcal{O}(d \log d)$ bound for $m \ge 9$. Furthermore, the variation in execution time across independent trials is negligible, underscoring the algorithm's highly consistent performance for massive polynomial degrees.

\begin{figure}[htbp]
    \centering
    \includegraphics[width=0.7\textwidth]{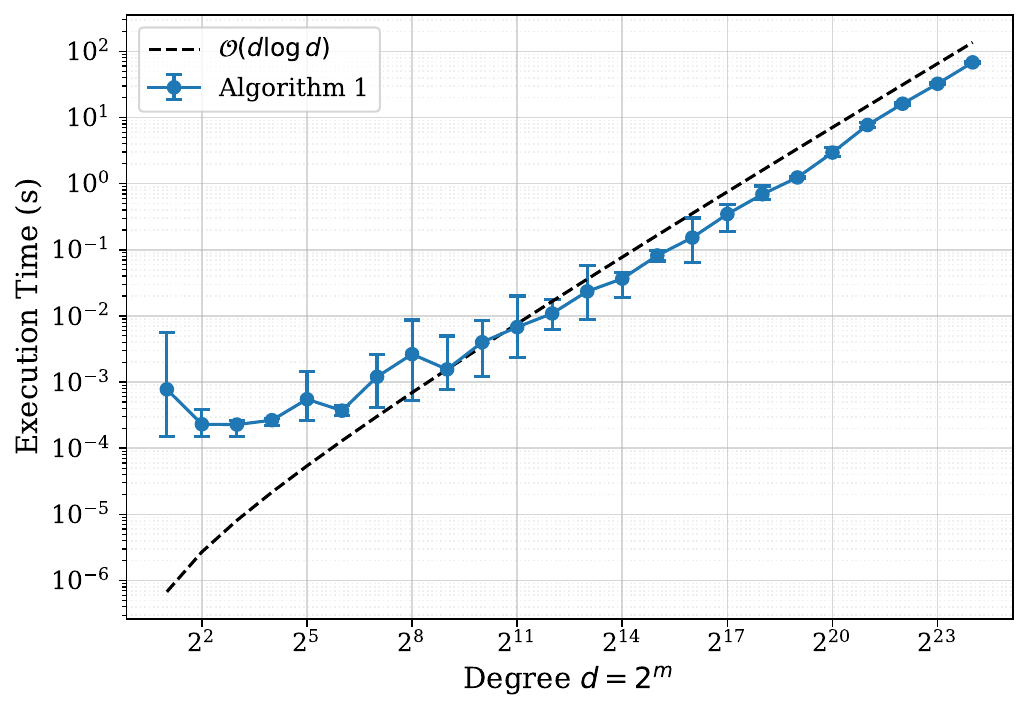}
    \caption{Execution time for computing the explicit UCR circuit parameters
in \cref{alg:qsp} as a function of the polynomial degree $d=2^m$. The solid blue line denotes the mean empirical execution time over $10$ independent random trials corresponding to random Laurent polynomials, with error bars indicating the absolute minimum and maximum times recorded for each degree. The dashed black line illustrates the theoretical $\mathcal{O}(d \log d)$ asymptotic scaling.}
    \label{fig:scaling}
\end{figure}

\begin{table}[htbp]
    \centering
    \caption{Overall complexity of the block-encoding quantum circuit for a matrix polynomial \eqref{eq:complexUBE}. Gate counts include the two UCRs for $U_{f_d,4d}^{(\mathbb C)}$ ($8d-2$ single-qubit gates, $8d-4$ CNOTs), the exact $\mathrm{QFT}_J$ and $\mathrm{QFT}_J^\dagger$ ($3m^2+11m+10$ single-qubit gates, $2m^2+6m+8+6\lfloor m/2 \rfloor$ CNOTs), auxiliary CNOT/SWAP gates ($4m+8$ CNOTs), and the components $P_{\mathrm{in}}$ and $P_{\mathrm{out}}$ ($2m+3$ single-qubit gates, $1$ CNOT). Controlled-$U$ queries are counted for $\mathcal{W}_U^{(S,J)}$ and its inverse. $D_{c\text{-}U}$ denotes the circuit depth to implement a given input model.}
    \label{tab:quantum_resources_summary}
    \begin{tabular}{|l|c|}
        \hline
        Resource & Total Count \\
        \hline
        Single-qubit gates & $2^{m+3} + 3m^2 + 13m + 11$ \\
        \hline
        CNOTs & $2^{m+3} + 2m^2 + 10m + 6\lfloor m/2 \rfloor + 13$ \\
        \hline
        $c\text{-}U$ queries & $2^{m+3}$ (non-commuting), $2m+4$ (commuting) \\
        \hline
        Circuit depth & $\mathcal{O}(dD_{c\text{-}U})$ (non-commuting), $\mathcal{O}(d+D_{c\text{-}U}\log d)$ (commuting) \\
        \hline
    \end{tabular}
\end{table}

 To validate \cref{thm:numerical_stability}, we
conduct numerical experiments measuring the
encoding error as a function of the polynomial
degree and bit precision, which is defined as follows
\begin{align}\label{eq: encoding error}
    \frac{\max_{E\in S}
    \abs{\widetilde f(e^{-iE})
    -f_d(e^{-iE})}}{\max_{E\in S}\abs{f_d(e^{-iE})}},
\end{align}for $S\subset[-\pi,\pi]$.
For the experiments, we let $S$ be the set of $21$ equally
spaced points in $[-\pi,\pi]$. Equivalently, we consider the Hamiltonian whose eigenvalues are equal to these equally spaced points. For our numerical tests, we randomly
generated complex Laurent polynomials
$f_d(z)=\sum_{k=-d}^{d}c_kz^k$ by independently sampling the real and imaginary parts of each raw coefficient from the uniform distribution on $[-1,1]$. We normalized the coefficient vector so that \(\sum_{k=-d}^{d}|c_k|=1\). We then computed $\eta$ as defined in \cref{cor:complex_diag_encoding} and used $\frac{f_d}{\eta}$ as the normalized polynomial. 

To account for numerical errors in the FFT, the FWHT, and elementary-function evaluations, we emulated reduced-precision arithmetic by rounding intermediate FFT and FWHT butterfly results and the computed polar angles to a prescribed number of bits. This rounding procedure is consistent with the proof of \cref{thm:numerical_stability}. Since the block encoding of matrix polynomial is expressed by the matrix-vector product corresponding to \eqref{eq:complexUBE} and \cref{def: block encoding}, we  applied the UCRs \eqref{eq:app_general_ucr} to the state vector via blockwise $2\times 2$ matrix–vector products. Each rotation block was computed by summing the angles $\widehat{\boldsymbol{\vartheta}}^{(y)}$ or $\widehat{\boldsymbol{\vartheta}}^{(z)}$ with signs determined by the CNOT sequence. For more  implementation details, we refer the reader to the code linked in the Code Availability section.

\Cref{fig:numerical_error_analysis} shows that
the encoding error exhibits an approximately
linear trend on a log-log scale as the
polynomial degree increases at fixed double
precision. We also observe that
an approximately exponential decrease in
the encoding error with increasing bit
precision, until the error approaches the
double-precision numerical floor. As can be deduced from the proof of \cref{thm:numerical_stability},
these observations are consistent with \eqref{ineq: last}.

\begin{figure}[htbp]
    \centering
    \includegraphics[width=0.48\textwidth]{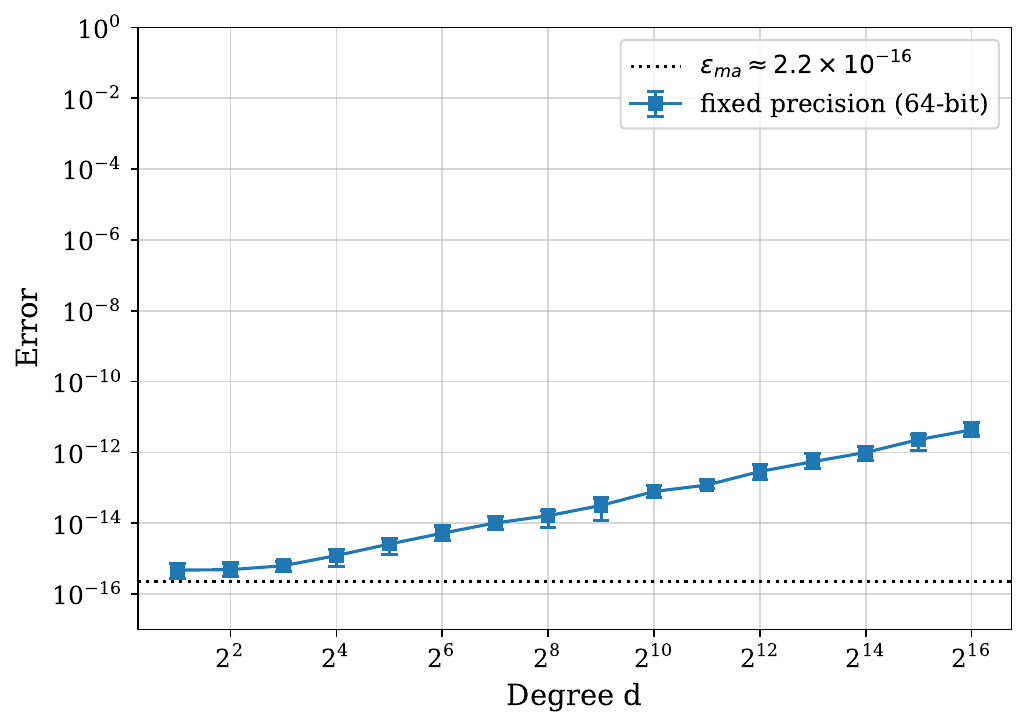}
\includegraphics[width=0.48\textwidth]{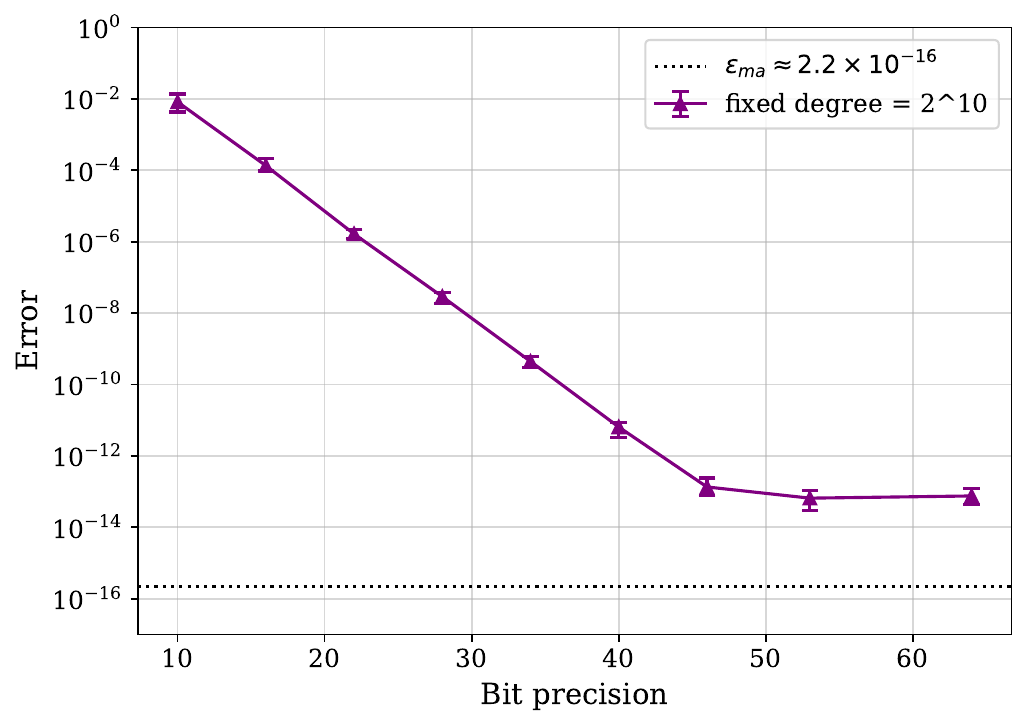}
    
    \caption{Simulations of the encoding error defined in \eqref{eq: encoding error} as a function of the polynomial degree and numerical bit precision. In both panels, solid lines show the mean empirical error over 10 independent trials using randomly generated Laurent polynomials, as described in the main text. Error bars indicate the minimum and maximum errors across the trials. The dotted black lines indicate machine precision ($\epsilon \approx 2.2 \times 10^{-16}$).}
    \label{fig:numerical_error_analysis}
\end{figure}

We also compare our method to a GQSP method in \cite{yamamoto2024robust} that uses Prony completion and
carving for Hamiltonian simulation. Both
methods block encode the same degree-$d$ Jacobi--Anger truncation
$P_d$ of $e^{-i\tau\cos\theta}$. We set $\tau=100$, the largest
simulation time considered in \cite{yamamoto2024robust}. Following their partitioning scheme for Prony completion, we write
$P_d=\beta(F_1+F_2)$, where
$F_1(z)=\gamma(z^d+z^{-d})$ and
$F_2(z)=P_d(z)/\beta-F_1(z)$. We set $\beta=2$ and $\gamma=0.1$ to satisfy the condition they impose,
$\max_{|z|=1}|F_j(z)|<1$ for $j=1,2$. With these fixed parameter
values, the GQSP method provides 
comparable results to those reported in \cite{yamamoto2024robust}
for $\tau=100$. They report the encoding error as a function of the total number of queries to controlled-$U$ and controlled-$U^\dagger$. For a Laurent polynomial of degree $d$, this count is $4d$. Here, we report the encoding error directly as a function of $d$.

Similar to the error metric in \cite{yamamoto2024robust}, we  define the error by
\begin{equation}
    \max_{\theta \in \Theta}
    \left|
        \widehat{P}_d(e^{i\theta}) - e^{-i\tau\cos\theta}
    \right|,
    \label{eq:total_simulation_error}
\end{equation}
where $\widehat{P}_d(e^{i\theta})$ denotes the circuit output
multiplied by its normalization factor and  
   \begin{align}
       \Theta =
    \left\{
        -\pi+\frac{2\pi}{N}(j+\sqrt{2}-1)
        \;\middle|\; j=0,\ldots,N-1
    \right\},
    \qquad N=1024.
   \end{align} 
We use the irrational shift $\sqrt{2}-1$ to avoid the FFT/UCR
interpolation nodes. Otherwise, evaluating the error only at those
nodes could miss errors between them.

\cref{fig:numerical_error_analysis 1} and \cref{tab:error_comparison} report numerical results. For $\tau=100$, the Jacobi-Anger truncation error dominates at $d=50$ and $d=100$, resulting in comparable performance for both methods. However, for $d\ge150$, the truncation error reaches the numerical precision floor. In this regime, our method consistently achieves a total error below $3\times10^{-14}$, outperforming the GQSP baseline by two to three orders of magnitude. We also compare these methods at $\tau=800$. This value is comparable to $\lambda_H t \approx 720$ for the
$16$-nucleon response calculation on an $8^3$ lattice using GQSP
\cite{spagnoli2026quantum}, where $t$ is the evolution time
corresponding to an energy resolution of $100\,\mathrm{MeV}$ and
$\lambda_H$ is the $\ell_1$-norm of the LCU coefficients. At $d=850$, the error of our method is approximately $4.7\times 10^{-7}$, tracking the Jacobi-Anger truncation error, and it remains near $5\times10^{-13}$  for $d\ge900$. However, the error of GQSP method  slowly decays from $1.3\times10^{-1}$ at $d=900$ to $7.0\times10^{-8}$ at $d=1050$.

\begin{figure}[htbp]
    \centering
    \includegraphics[width=0.48\textwidth]{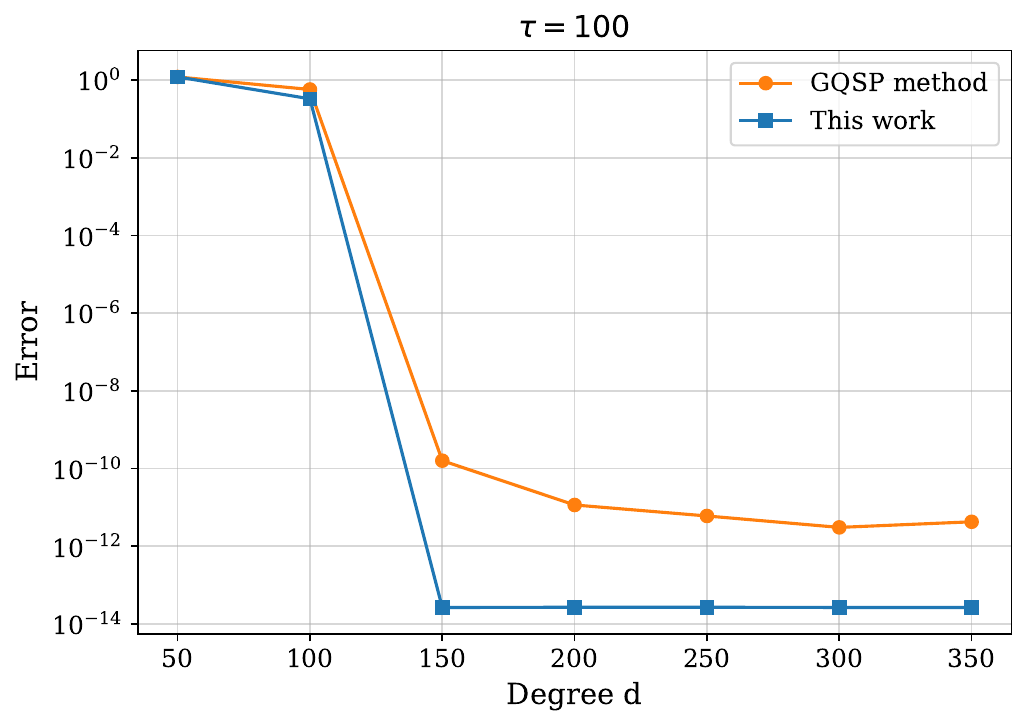}
\includegraphics[width=0.48\textwidth]{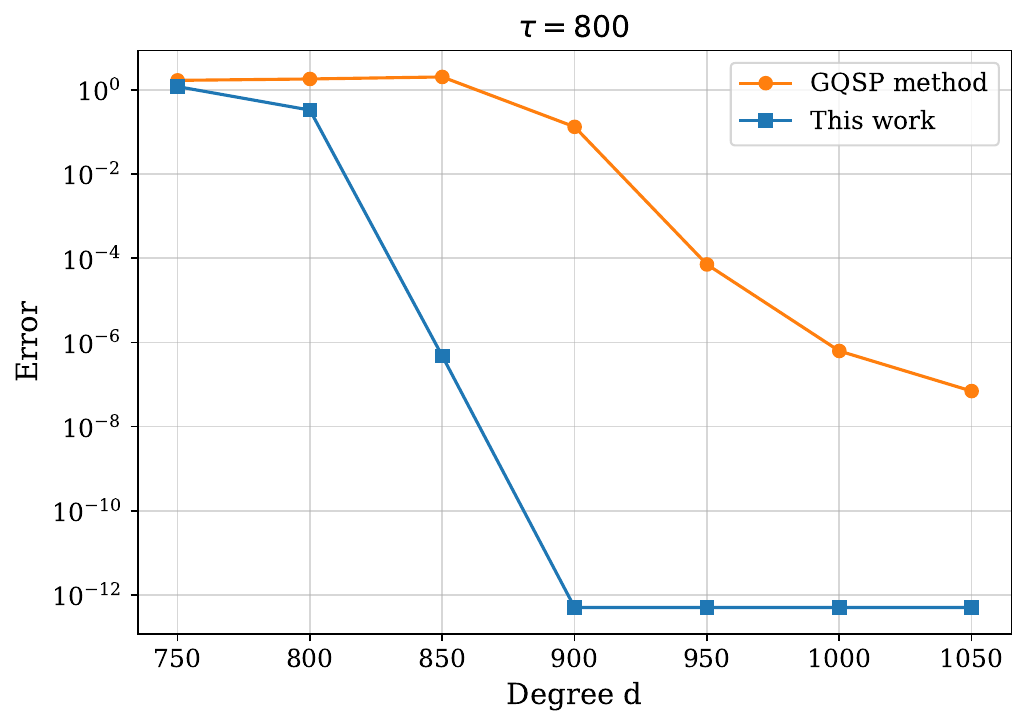}
    
    \caption{Comparison of our algorithm to the GQSP method with Prony completion and carving~\cite{yamamoto2024robust}  
    in terms of the encoding error defined in \eqref{eq: encoding error} as a function of the Laurent polynomial degree. }
\label{fig:numerical_error_analysis 1}
\end{figure}

\begin{table}[htbp]
    \centering
    \caption{Comparison of methods for $\tau = 100$ and $\tau = 800$ in the simulation results in \cref{fig:numerical_error_analysis 1}. ``Truncation'' denotes the truncation error from the Jacobi--Anger expansion. For both methods, we report the  error defined in \eqref{eq:total_simulation_error}.}
    \label{tab:error_comparison}
    \begin{tabular}{|c|c|c|c|}
        \hline
        Degree ($d$) & Truncation & GQSP \cite{yamamoto2024robust} & This work \\
        \hline
        \multicolumn{4}{|c|}{$\tau = 100$} \\
        \hline
        50  & $1.204 \times 10^{0}$   & $1.204 \times 10^{0}$   & $1.204 \times 10^{0}$   \\
        \hline
        100 & $3.290 \times 10^{-1}$  & $5.726 \times 10^{-1}$  & $3.290 \times 10^{-1}$  \\
        \hline
        150 & $3.427 \times 10^{-14}$ & $1.595 \times 10^{-10}$ & $2.649 \times 10^{-14}$ \\
        \hline
        200 & $3.422 \times 10^{-14}$ & $1.154 \times 10^{-11}$ & $2.690 \times 10^{-14}$ \\
        \hline
        250 & $3.422 \times 10^{-14}$ & $6.017 \times 10^{-12}$ & $2.690 \times 10^{-14}$ \\
        \hline
        300 & $3.422 \times 10^{-14}$ & $3.058 \times 10^{-12}$ & $2.656 \times 10^{-14}$ \\
        \hline
        350 & $3.422 \times 10^{-14}$ & $4.260 \times 10^{-12}$ & $2.656 \times 10^{-14}$ \\
        \hline
        \multicolumn{4}{|c|}{$\tau = 800$} \\
        \hline
        750  & $1.200 \times 10^{0}$   & $1.689 \times 10^{0}$   & $1.200 \times 10^{0}$   \\
        \hline
        800  & $3.319 \times 10^{-1}$  & $1.822 \times 10^{0}$   & $3.319 \times 10^{-1}$  \\
        \hline
        850  & $4.737 \times 10^{-7}$  & $2.036 \times 10^{0}$   & $4.737 \times 10^{-7}$  \\
        \hline
        900  & $5.346 \times 10^{-13}$ & $1.327 \times 10^{-1}$  & $4.997 \times 10^{-13}$ \\
        \hline
        950  & $5.355 \times 10^{-13}$ & $7.116 \times 10^{-5}$  & $4.996 \times 10^{-13}$ \\
        \hline
        1000 & $5.346 \times 10^{-13}$ & $6.244 \times 10^{-7}$  & $4.996 \times 10^{-13}$ \\
        \hline
        1050 & $5.358 \times 10^{-13}$ & $6.959 \times 10^{-8}$  & $5.004 \times 10^{-13}$ \\
        \hline
    \end{tabular}
\end{table}

\subsection{Time-evolution circuit as an input model}

\cref{thm:chebyshevQSP} is applicable to the case of real polynomials without parity constraint, and requires controlled block encodings of Hamiltonian matrices. However, it is often the case that constructing just a block encoding of a Hamiltonian requires many ancilla qubits and multi-controlled operations \cite{sunderhauf2024block,camps2024explicit}. In contrast, the controlled Trotter formula, as a specialized input model of \cref{thm:complexQSP}, demands only a single ancilla qubit and can be implemented with single- and two- qubit gate operations, thereby being suitable for early fault-tolerant applications. This advantage is highlighted in recent  applications including ground-state property estimation \cite{dong2022ground,chakraborty2025quantum} and thermal property estimation \cite{ko2026particle}.  Furthermore, using the time-evolution circuit as input model provides efficient implementations of non-trivial Hamiltonian simulations (e.g. square-root phase function and fractional time evolution \cite{motlagh2024generalized}).

With this in mind, we discuss how to apply \cref{thm:complexQSP}
for the task of estimating observables with $U_T$ being approximated by Trotter circuits. Specifically, we consider estimating $\tr(\rho_0 f_d(e^{iH})^\dag Of_d(e^{iH}))$ where $\rho_0$ is the initial state and $O$ is an observable.  To estimate this quantity to precision $\epsilon$, we can apply the extrapolation technique \cite{chakraborty2025quantum} and achieve a poly-logarithmic Trotter circuit depth with respect to $\epsilon$, although their circuit construction resembles the conventional GQSP circuit.   

Suppose that our scheme and the one in \cite{chakraborty2025quantum} use the same Trotter formula for the time-evolution input model $U_T$. Then, by the Baker–Campbell–Hausdorff formula, there exists an effective Hamiltonian $\widetilde H$ such that
\begin{align}
    \widetilde U_T= e^{i\widetilde H}.
\end{align}By using this fact, both approaches return the same expectation value, $\tr(\rho_0 f_d(e^{i \widetilde H})^\dag Of_d(e^{i\widetilde H}))$. From this observation, we borrow their result \cite[Theorem 3.4]{chakraborty2025quantum} where the Richardson extrapolation technique is applied to the expectation values associated to preassigned Trotter step sizes. This allows us to achieve a poly-logarithmic Trotter circuit depth with respect to precision $\epsilon$ for estimating the observable, that is, $\tr(\rho_0 f_d(e^{i \widetilde H})^\dag Of_d(e^{i\widetilde H}))\pm\epsilon\norm{O}_2$.

\section*{Code Availability}
Code used for the current study is available at the following GitHub repository: 

https://github.com/TKmath/Exact-and-Efficient-Circuit-Construction-for-Block-Encoding-Matrix-Polynomials/tree/main

\section*{Data Availability}
All data that support the findings of this study are included within the Github repository in the Code Availability section.

\section*{Acknowledgements}
The author acknowledges support from a KIAS Individual Grant (CG096001) at the Korea Institute for Advanced Study. 

\appendix

\section{Review of QSP without angle finding}

An interpolation-based formulation of QSP introduced by Alase
\cite{alase2025quantum} avoids finding phase factors and resembles LCU in
that it directly encodes information about the target function into a
quantum circuit. As a crucial component, this approach assumes a diagonal block encoding of function values at specialized interpolation points as restated as follows. 

\begin{definition}[Diagonal encoding of a function]
Let $f:S^1\to\mathbb{C}$ be a function such that $\|f\|_\infty \le 1$. For $d=2^m$ with $m\in\mathbb{Z}^+$, we say that an $(m+3)$-qubit unitary operator $U_{f,4d}$ encodes $f$ if
\begin{equation}
\label{eq:diag_encoding_function}
    \bra{j'}_J\bra{0}_a
    U_{f,4d}
    \ket{j}_J\ket{0}_a
    =
    \delta_{jj'}
    f\left(e^{i2\pi j/4d}\right),
    \qquad
    j,j'\in\{0,1,\ldots,4d-1\},
\end{equation}where $J$ denotes an $(m+2)$-qubit register and $a$ corresponds to a single ancilla qubit.
\end{definition}

The choice of these interpolation points is due to the following lemma.

\begin{lemma}[{\cite[Lemma 8]{alase2025quantum}}]
\label{lem:laurentpolyapprox}
Let $f_d:S^1 \to \mathbb{C}$ be a Laurent polynomial of degree $d$, and let 
$\{z_k = \exp\left(2\pi ik/4d\right) \mid k=0,\dots,4d-1\}$ be a set of interpolation points. Then, 
\begin{equation}
\label{eq:polyapprox}
        f_d(z) = \frac{1}{8d^2}\sum_{j'=d}^{3d-1}\sum_{j,k=0}^{4d-1} f_d(z_k)\left(\frac{z}{z_k}\right)^{j-j'}.
\end{equation}
\end{lemma}

Based on this expansion, the original work constructs an $(m+3)$-qubit
block encoding of $f_d(U)$ as follows.

\begin{lemma}[{\cite[Lemma 9]{alase2025quantum}}]\label{lem: lemma 9}
Let $f_d$ be a Laurent polynomial of degree $d$, and define the states
\begin{equation}
    \ket{+_{2d}}_J
    =
    \frac{1}{\sqrt{2d}}
    \sum_{j=d}^{3d-1} \ket{j}_J,
    \qquad
    \ket{+_{4d}}_J
    =
    \frac{1}{\sqrt{4d}}
    \sum_{j=0}^{4d-1} \ket{j}_J.
    \label{eq:plus-states}
\end{equation}
Then,
\begin{equation}
    f_d(U)
    =
    \sqrt{2}\,
    \bra{+_{2d}}_J
    \mathcal{W}_U^\dagger
    \bigl[\mathbb{I}\otimes V_{4d}\bigr]
    \mathcal{W}_U
    \ket{+_{4d}}_J,
\end{equation}
where
\begin{align*}
    \mathcal{W}_U &= \sum_{j=0}^{4d-1}U^j\otimes \ketbra{j}, \\
    V_{4d} &= \operatorname{QFT}_J(\bra{0}U_{f,4d}\ket{0})\operatorname{QFT}_J^\dag.
\end{align*}
\end{lemma}
Here, we restate the result of \cite{alase2025quantum} by explicitly specifying the maximum degree of $U$ in $\mathcal{W}_U$ as $4d-1$, which is justified by \cite[Lemma 7]{alase2025quantum}. 


\section{Review of Uniformly Controlled Rotations}
\label{app:ucr_details}

A uniformly controlled single-qubit rotation can be implemented
using only CNOT gates and single-qubit rotations
\cite{mottonen2004quantum}.
The corresponding rotation angles can be obtained efficiently using the
fast Walsh-Hadamard transform (FWHT). Related applications can be found in
\cite{amankwah2022quantum,nakaji2022approximate}. 

In the following, we briefly review this procedure.

Let $J$ be an $M$-qubit control register (in our case, $M = m+2$, such that $2^M = 4d$) and let $a$ denote a single target qubit. Let $g(x)=x\oplus \lfloor x/2\rfloor$ ($x=0,\ldots,2^M-1$) denote the binary-reflected Gray-code map from \cite{mottonen2004quantum}, and let $P_{G}$ be the
corresponding permutation matrix defined by $(P_{G}\boldsymbol{v})_x=v_{g(x)}$.
For $\nu\in\{y,z\}$, a decomposition of the corresponding UCR is formulated as
\begin{equation}
\label{eq:app_general_ucr}
    F^{(\nu)}_{J,a}
    (\boldsymbol{\alpha}^{(\nu)})
    :=
    \sum_{x=0}^{2^M-1}
    \ketbra{x}_J
    \otimes
    R_\nu\left(
        \alpha_x^{(\nu)}
    \right)_a=
    \prod_{q=0}^{2^M-1}
    \left[
        \operatorname{CNOT}_{J_{c_q}\rightarrow a}
        R_\nu\left(
            \vartheta_q^{(\nu)}
        \right)_a
    \right].
\end{equation}
Here, $\boldsymbol{\alpha}^{(\nu)}
    =
    \begin{pmatrix}
        \alpha_0^{(\nu)} &
        \alpha_1^{(\nu)} &
        \cdots &
        \alpha_{2^M-1}^{(\nu)}
    \end{pmatrix}^{\mathsf T}$, and the product in \eqref{eq:app_general_ucr} is understood in the
left-to-right circuit execution order. From the mapping $g(x)$, we identify $c_q$ as the control qubit on
which $g(q)$ and $g((q+1)\bmod 2^M)$ differ. Similarly, we define $\boldsymbol{\vartheta}^{(\nu)}=\begin{pmatrix}
        \vartheta_0^{(\nu)} &
        \vartheta_1^{(\nu)} &
        \cdots &
        \vartheta_{2^M-1}^{(\nu)}
    \end{pmatrix}^{\mathsf T}$. We define the unnormalized Walsh--Hadamard matrix $H_{2^M}$ by $(H_{2^M})_{x,y}=(-1)^{x\cdot y}$,
where $x\cdot y$ denotes the modulo-$2$ inner product of the binary
representations of $x$ and $y$. 

For $\nu \in \{y, z\}$, by identifying $\boldsymbol{\alpha}^{(y)} = \boldsymbol{\beta}$ and $\boldsymbol{\alpha}^{(z)} = \boldsymbol{\gamma}$ from \eqref{eq:general_angles}, the result in \cite{mottonen2004quantum} states that
\begin{equation}\label{eq:thetas}
    \boldsymbol{\vartheta}^{(\nu)}
    =
    \frac{1}{2^M}
    P_{G}H_{2^M}\boldsymbol{\alpha}^{(\nu)}.
\end{equation}
Using the FWHT, we compute the matrix-vector product $H_{2^M}\boldsymbol{\alpha}^{(\nu)}$ in $\mathcal{O}(2^M \log 2^M)$ time. Subsequently, applying the permutation matrix $P_{G}$ to the resulting vector yields $\boldsymbol{\vartheta}^{(\nu)}$, which requires an additional $\mathcal{O}(2^M)$ time. Since computing $\boldsymbol{\alpha}^{(\nu)}$ from \eqref{eq:general_angles} takes $\mathcal{O}(2^M)$ time, and $2^M = 4d$, the total time complexity to compute the exact explicit circuit parameters is $\mathcal{O}(d\log d)$.

As noted in \cite{mottonen2004quantum}, we can choose a horizontally-mirrored Gray-code decomposition for
$F^{(z)}_{J,a}$ such that its first CNOT coincides with the final CNOT of 
$F^{(y)}_{J,a}$. This is illustrated in
\cref{fig:ucr_boundary_cancellation}. 
Specifically, we write 
\begin{align}
    &F^{(y)}_{J,a}
    \left[
        \boldsymbol{\alpha}^{(y)}
    \right]=\prod_{q=0}^{2^M-1}
    \left[
        \operatorname{CNOT}_{J_{c_q}\rightarrow a}
        R_y\left(
            \vartheta_q^{(y)}
        \right)_a
    \right],
    \\
    &F^{(z)}_{J,a}
    \left[
        \boldsymbol{\alpha}^{(z)}
    \right]=\prod_{q=2^M-1}^{0}
    \left[
        R_z\left(
            \vartheta_q^{(z)}
        \right)_a \operatorname{CNOT}_{J_{c_q}\rightarrow a}
    \right].
\end{align}We use the formula \eqref{eq:thetas} to find $\vartheta_q^{(y)}$ and $\vartheta_q^{(z)}$.

Consequently, the diagonal block encoding, 
$F^{(z)}_{J,a}(\boldsymbol{\gamma})
 F^{(y)}_{J,a}(\boldsymbol{\beta})$, in \cref{lem:general_diagonal_encoding}
requires $2^{M+1}-2$ CNOT gates rather than $2^{M+1}$, while the number of
single-qubit rotations remains $2^{M+1}$. The cyclic shift only changes the
Gray-code permutation in \eqref{eq:thetas} and does not affect the
$\mathcal{O}(d\log d)$ classical compilation complexity.

\begin{figure}[htbp]
    \centering

    \makebox[\textwidth][c]{
    \Qcircuit @C=0.75em @R=1.3em {
        & \lstick{J_{c_*}}
        & \qw
        & \qw
        & \ctrl{1}
        & \ctrl{1}
        & \qw
        & \qw
        \\
        & \lstick{a}
        & \push{\cdots}\qw
        & \gate{R_y(\vartheta^{(y)}_{2^M-1})}
        & \targ
        & \targ
        & \gate{R_z(\vartheta^{(z)}_{0})}
        & \push{\cdots}\qw
    }
    }

    \vspace{0.7em}

    \[
        \Downarrow\qquad
        \operatorname{CNOT}_{J_{c_*},a}^2=\mathbb{I}
    \]

    \vspace{0.4em}

    \makebox[\textwidth][c]{
    \Qcircuit @C=0.75em @R=1.3em {
        & \lstick{J_{c_*}}
        & \qw
        & \qw
        & \qw
        & \qw
        & \qw
        \\
        & \lstick{a}
        & \push{\cdots}\qw
        & \gate{R_y(\vartheta^{(y)}_{2^M-1})}
        & \gate{R_z(\vartheta^{(z)}_{0})}
        & \push{\cdots}\qw
        & \qw
    }
    }

    \caption{An illustration of
    the cancellation of the two CNOT gates at the boundary between
    $F^{(y)}_{J,a}(\boldsymbol{\beta})$ and
    $F^{(z)}_{J,a}(\boldsymbol{\gamma})$.
    The cyclic Gray-code starting point of the second UCR is chosen
    such that its first Gray transition flips the same control qubit
    $J_{c_*}$ as the final Gray transition of the first UCR.
    The resulting adjacent CNOT gates are identical and therefore
    cancel.
    }
    \label{fig:ucr_boundary_cancellation}
\end{figure}

\section{The proof of Theorem 3}
\label{app:proof thm3}

Here we provide the proof of \cref{thm:numerical_stability}.
\begin{proof}
Throughout the proof, we use the hat notation $\widehat{\cdot}$ to denote quantities computed in floating-point arithmetic and $C_i$ to denote absolute constants independent of the polynomial degree and machine precision.

For Step 1 of \cref{alg:qsp}, we can write the true function values via the following relation
\begin{align}
    \boldsymbol{v}= F_{4d} \boldsymbol{c}, 
\end{align}where
\begin{align}
    \boldsymbol{v}= \left(f_d(z_0),...,f_d(z_{4d-1}) \right)^T,\quad \boldsymbol{c}
    :=\big(
    c_0,\ldots,c_d,
    \underbrace{0,\ldots,0}_{2d-1},
    c_{-d},\ldots,c_{-1}
    \big)^T, \quad (F_{4d})_{k,j}=e^{2\pi i kj/4d}.
\end{align}By  applying \cite[Theorem 24.2]{higham2002accuracy}, an approximation of $\boldsymbol{v}$, $\widehat{\boldsymbol{v}}$, can be obtained using the standard radix 2 FFT, satisfying 
\begin{align}
    \norm{\widehat{\boldsymbol{v}}-\boldsymbol{v}}_2\le \frac{\log_2(4d)\delta \norm{\boldsymbol{v}}_2}{1-\log_2(4d)\delta}.
\end{align}Here $\delta=c\epsilon_{ma}$, where $\epsilon_{ma}$ denotes the machine precision and $c$ is a constant dependent on methods that computes the unity of circle (See \cite[eq (24.2)]{higham2002accuracy}). 
Using the equivalence of $\infty$-norm and $2$-norm that $\norm{\boldsymbol{x}}_\infty \le \norm{\boldsymbol{x}}_2\le \sqrt{4d}\norm{\boldsymbol{x}}_\infty$, we have
\begin{align}
    \norm{\widehat{\boldsymbol{v}}-\boldsymbol{v}}_\infty\le \frac{2\log_2(4d)\sqrt{d}\delta \eta}{1-\log_2(4d)\delta},
\end{align}by recalling that $\eta=\max_k \abs{f_d(z_k)}=\norm{\boldsymbol{v}}_\infty$ in \cref{cor:complex_diag_encoding}. In the actual implementation of our algorithm,
we compute the floating-point quantities
$\widehat\eta$ and $\widehat\lambda_k$ for the normalized function values, such that
\begin{align}\label{ineq: eta-v}
    \abs{\widehat\eta-\norm{\widehat{\boldsymbol v}}_\infty}
    \le C_0\epsilon_{ma}\norm{\widehat{\boldsymbol v}}_\infty,\quad    \abs{
    \widehat\lambda_k
    -\frac{\widehat v_k}{\widehat\eta}}
    \le C_1\epsilon_{ma},
\end{align}where $C_0\epsilon_{ma}<1$. For these values, we can estimate the error bound
\begin{align}\label{ineq: lambda}
    \abs{\widehat\lambda_k-\lambda_k}
    &\le \abs{\widehat\lambda_k - \frac{\widehat v_k}{\widehat \eta}} +
    \abs{
    \frac{\widehat v_k}{\widehat\eta}
    -\frac{\widehat v_k}{\norm{\widehat{\boldsymbol v}}_\infty}}
    +
    \abs{\frac{\widehat v_k}{\norm{\widehat{\boldsymbol v}}_\infty}-\frac{\widehat v_k}{\eta}
    }+\abs{\frac{\widehat v_k}{\eta}-\frac{v_k}{\eta}
    }\nonumber\\
    &\le C_1\epsilon_{ma}+ 
    \abs{\frac{\widehat\eta-\norm{\widehat{\boldsymbol v}}_\infty}
    {\widehat\eta}}
    +\frac{2\norm{\widehat{\boldsymbol{v}}-\boldsymbol{v}}_\infty}{\eta}\quad \text{by the triangle inequality}\nonumber\\
    &\le  C_1\epsilon_{ma}+\frac{C_0\epsilon_{ma}}
    {1-C_0\epsilon_{ma}}+ \frac{2\norm{\widehat{\boldsymbol{v}}-\boldsymbol{v}}_\infty}{\eta}\nonumber\\
    &\le C_2\sqrt{d} \log d\,\epsilon_{ma}.
\end{align}Subsequently, we observe that
\begin{align}
    \abs{\widehat\lambda_k}&\le \abs{\frac{\widehat v_k}{\widehat \eta}}+C_1\epsilon_{ma}\le \frac{1}{1-C_0\epsilon_{ma}}+C_1\epsilon_{ma}=1+\mathcal{O}(\epsilon_{ma}).
\end{align}This implies that for each computed normalized value
$\widehat\lambda_k$, we can define the floating-point quantity of its magnitude clipped in $[0,1]$,
\begin{align}
    \widehat{\abs{\widehat \lambda_k}}:=\min\{1,\operatorname{fl}(\abs{\widehat \lambda_k})\}
\end{align}and this clipping introduces only $\mathcal{O}(\epsilon_{ma})$ additive error,
\begin{align}
    \abs{\widehat{\abs{\widehat \lambda_k}}-\abs{\widehat \lambda_k}}
    \le C_3\epsilon_{ma}.
    \label{eq:computed_radius_error}
\end{align}
Let $\widehat\beta_k$ and
$\widehat\gamma_k$ denote the computed angles, which correspond to the definition \eqref{eq:general_angles}. For $\widehat\lambda_k\neq0$, we assume that the trigonometric
functions are evaluated to precision  $\epsilon_{ma}$,
\begin{align}
    \abs{
    \widehat\beta_k
    -2\arccos \widehat{\abs{\widehat \lambda_k}}}
    \le C_4\epsilon_{ma},\quad\abs{
    \widehat\gamma_k
    +2\operatorname{Arg}(\widehat\lambda_k)}
    \le C_5\epsilon_{ma}.
\end{align}
When $\widehat\lambda_k=0$, we choose $\widehat \beta_k=\pi$ and
$\widehat\gamma_k=0$. By using the smoothness of cosine function,  \eqref{eq:computed_radius_error}, and \eqref{ineq: lambda}, we obtain
\begin{align}
    \abs{
    \cos\left(\frac{\widehat\beta_k}{2}\right)
    -\abs{\widehat\lambda_k}}\le
    \frac12\abs{
    \widehat\beta_k
    -2\arccos \widehat{\abs{\widehat \lambda_k}}}+\abs{\widehat{\abs{\widehat \lambda_k}}-\abs{\widehat\lambda_k}}
    \le C_6\epsilon_{ma}.
    \label{eq:polar_radius_error}
\end{align}
Using the numeric inequality $\abs{e^{ix}-e^{iy}}\le\abs{x-y}$,
we obtain
\begin{align}
    \abs{
    e^{-i\widehat\gamma_k/2}
    -e^{i\operatorname{Arg}(\widehat\lambda_k)}
    }\le
    \frac12
    \abs{\widehat\gamma_k+2\operatorname{Arg}(\widehat\lambda_k)}\le \frac{C_5}{2}\epsilon_{ma}.
    \label{eq:polar_phase_error}
\end{align}
Combining \eqref{eq:polar_radius_error} and
\eqref{eq:polar_phase_error}, we have
\begin{align}
    \abs{
    e^{-i\widehat\gamma_k/2}
    \cos\left(\frac{\widehat\beta_k}{2}\right)
    -\widehat\lambda_k}\le
    \abs{\cos\left(\frac{\widehat\beta_k}{2}\right)
    -\abs{\widehat\lambda_k}}+
    \abs{\widehat\lambda_k}\abs{
    e^{-i\widehat\gamma_k/2}
    -e^{i\operatorname{Arg}(\widehat\lambda_k)}}\le C_7\epsilon_{ma}.
\end{align} 
By using this result
and \eqref{ineq: lambda}, it follows
\begin{align}
    \abs{
    e^{-i\widehat\gamma_k/2}
    \cos\left(\frac{\widehat\beta_k}{2}\right)
    -\lambda_k}\le
    \abs{
    e^{-i\widehat\gamma_k/2}
    \cos\left(\frac{\widehat\beta_k}{2}\right)
    -\widehat\lambda_k}
    +\abs{\widehat\lambda_k-\lambda_k}\le
    C_8\sqrt d\log(4d)\,\epsilon_{ma}.
    \label{eq:polar_block_error}
\end{align}For $\widehat\lambda_k=0$, the same bound follows
directly from the choice
$\widehat\beta_k=\pi$ and $\widehat\gamma_k=0$. 

We next analyze the error introduced during the procedure in \cref{app:ucr_details}.
We define the exact rotation-angle vectors
\begin{align}
    \boldsymbol\vartheta^{(y)}
    :=
    \frac{1}{4d} P_G^{(y)}H_{4d}\boldsymbol\beta,\quad
    \boldsymbol\vartheta^{(z)}
    :=
    \frac{1}{4d} P_G^{(z)}H_{4d}\boldsymbol\gamma,
\end{align}
where $H_{4d}$ and $P_G^{(\nu)}$ are defined in \eqref{eq:thetas}. Here the true rotation-angle vectors $\boldsymbol{\beta}$ and $\boldsymbol{\gamma}$, which are obtained from $\lambda_k$'s without floating-point errors, correspond to the approximate ones defined by $\widehat \beta_k$'s and $\widehat\gamma_k$'s.

Let
$\widehat{\boldsymbol\vartheta}^{(y)}$
and
$\widehat{\boldsymbol\vartheta}^{(z)}$
be the floating-point outputs of the FWHT applied
to $\widehat{\boldsymbol\beta}$ and
$\widehat{\boldsymbol\gamma}$, respectively.
For convenience, we write $\widehat{\boldsymbol\alpha}^{(y)}
    :=\widehat{\boldsymbol\beta}$ and $\widehat{\boldsymbol\alpha}^{(z)}
    :=\widehat{\boldsymbol\gamma}$. Similar to the proof of \cite[Theorem 24.2]{higham2002accuracy}, the FWHT satisfies
\begin{align}
\norm{    \widehat{\boldsymbol\vartheta}^{(\nu)}
    -
    \frac{1}{4d} P_G^{(\nu)}H_{4d}
    \widehat{\boldsymbol\alpha}^{(\nu)}}_2\le
    \frac{\log_2(4d)\epsilon_{ma}}{1-\log_2(4d)\epsilon_{ma}}\frac{\norm{\widehat{\boldsymbol\alpha}^{(\nu)}}_2}{\sqrt{4d}}.
    \label{eq:fwht_error}
\end{align}
Since $\norm{\widehat{\boldsymbol\alpha}^{(\nu)}}_2
    \le 6\pi\sqrt d$ as $\abs{\alpha_k^{(\nu)}}\le 2\pi+\mathcal{O}(\epsilon_{ma})$ by definition, we obtain
\begin{align}
    &\norm{ \widehat{\boldsymbol\vartheta}^{(\nu)}
    -
    \frac{1}{4d} P_G^{(\nu)}H_{4d}
    \widehat{\boldsymbol\alpha}^{(\nu)}}_1\le
    C_{9}\sqrt d\log d\,\epsilon_{ma},
    \label{eq:fwht_l1_error}
\end{align}by using the inequality that $\norm{\boldsymbol{x}}_1\le \sqrt{4d}\norm{\boldsymbol{x}}_2$.

By \eqref{eq:app_general_ucr}, we consider the UCRs that take as input $\boldsymbol{\vartheta}^{(\nu)}$ and $\boldsymbol{\widehat\vartheta}^{(\nu)}$. Using \eqref{eq:U_f_4d}, we can define exact and approximate diagonal-encoding circuits as follows, 
\begin{align}
    &U_f = \mathcal{U}(\boldsymbol{\vartheta}^{(y)},\boldsymbol{\vartheta}^{(z)}),\\
    &\widetilde U_f = \mathcal{U}(\frac{1}{4d} P_G^{(y)}H_{4d}
    \widehat{\boldsymbol\alpha}^{(y)}, \frac{1}{4d} P_G^{(z)}H_{4d}
    \widehat{\boldsymbol\alpha}^{(z)}),\\
    &\widehat U_f = \mathcal{U}(\widehat{\boldsymbol\vartheta}^{(y)},\widehat{\boldsymbol\vartheta}^{(z)}),
\end{align}where the parameterized unitary circuit $\mathcal{U}$ corresponds to the formulation \eqref{eq:U_f_4d} with CNOTs and rotations.

Using the numeric inequality for Pauli rotations that for any real $\theta,\delta$ and
$\nu\in\{y,z\}$, 
\begin{align}
    \norm{R_\nu(\theta+\delta)-R_\nu(\theta)}_2
    &=
    2\abs{\sin\frac{\delta}{4}}
    \le\frac{\abs{\delta}}{2},
    \label{eq:rotation_lipschitz}
\end{align}and a telescoping trick for products
of unitary matrices (e.g. \cite[eq (4.63)]{nielsen2000quantum}), together with
\eqref{eq:fwht_l1_error}, we obtain
\begin{align}
    \norm{
    \widehat U_f
    -\widetilde U_f
    }_2\le
    \frac12
    \sum_{\nu\in\{y,z\}}
    \norm{
    \widehat{\boldsymbol\vartheta}^{(\nu)}
    -
    \frac{1}{4d} P_G^{(\nu)}H_{4d}
    \widehat{\boldsymbol\alpha}^{(\nu)}}_1\le C_{10}\sqrt d\log d\,\epsilon_{ma}.
    \label{eq:ucr_rounding_error}
\end{align}
Now we estimate the error between the diagonal matrices $(\mathbb{I}_J\otimes \bra{0}_a) U_f (\mathbb{I}_J\otimes \ket{0}_a)$ and $(\mathbb{I}_J\otimes \bra{0}_a) \widehat U_f (\mathbb{I}_J\otimes \ket{0}_a)$, which we obtain by using \eqref{eq:general_direct_complex_diag_block}. Notice that
\begin{align}
    &\norm{(\mathbb{I}_J\otimes \bra{0}_a) \widehat U_f (\mathbb{I}_J\otimes \ket{0}_a)-(\mathbb{I}_J\otimes \bra{0}_a) U_f (\mathbb{I}_J\otimes \ket{0}_a)}_2\nonumber\\
    &\le \norm{(\mathbb{I}_J\otimes \bra{0}_a) \widehat U_f (\mathbb{I}_J\otimes \ket{0}_a)-(\mathbb{I}_J\otimes \bra{0}_a)\widetilde U_f (\mathbb{I}_J\otimes \ket{0}_a)}_2 \nonumber\\
    &+ \norm{(\mathbb{I}_J\otimes \bra{0}_a) \widetilde U_f (\mathbb{I}_J\otimes \ket{0}_a)-(\mathbb{I}_J\otimes \bra{0}_a) U_f (\mathbb{I}_J\otimes \ket{0}_a)}_2\nonumber\\
    &\le \norm{\widehat U_f-\widetilde U_f}_2 \nonumber\\
    &+ \norm{(\mathbb{I}_J\otimes \bra{0}_a) \widetilde U_f (\mathbb{I}_J\otimes \ket{0}_a)-(\mathbb{I}_J\otimes \bra{0}_a) U_f (\mathbb{I}_J\otimes \ket{0}_a)}_2.
\end{align}Here the second term in RHS simplifies to
\begin{align}
    \norm{(\mathbb{I}_J\otimes \bra{0}_a) \widetilde U_f (\mathbb{I}_J\otimes \ket{0}_a)-(\mathbb{I}_J\otimes \bra{0}_a) U_f (\mathbb{I}_J\otimes \ket{0}_a)}_2=\max_{0\le k<N}
    \abs{
    e^{-i\widehat\gamma_k/2}
    \cos\left(\frac{\widehat\beta_k}{2}\right)
    -
    \lambda_k}.
\end{align}Thus, by using \eqref{eq:polar_block_error} and
\eqref{eq:ucr_rounding_error}, we have
\begin{align}\label{ineq: u-u}
    \norm{(\mathbb{I}_J\otimes \bra{0}_a) \widehat U_f (\mathbb{I}_J\otimes \ket{0}_a)-(\mathbb{I}_J\otimes \bra{0}_a) U_f (\mathbb{I}_J\otimes \ket{0}_a)}_2\le C_{11}\sqrt d\log d\,\epsilon_{ma}.
\end{align}Finally, based on the construction \eqref{eq:complexUBE}, we consider the block encodings of matrix polynomials obtained from exact and approximate rotation angles, 
\begin{align}
    &U_{BE}:=
    \left(
    \mathbb{I}_S\otimes P_{\mathrm{out}}^\dagger\otimes \mathbb{I}_{a}
    \right)
    \left(
    \mathcal W_U^{(S,J)}\otimes \mathbb{I}_{a}
    \right)^\dagger
    \left(
    \mathbb{I}_S\otimes U_{\mathrm{diag}}
    \right)
    \left(
    \mathcal W_U^{(S,J)}\otimes \mathbb{I}_{a}
    \right)
    \left(
    \mathbb{I}_S\otimes P_{\mathrm{in}}\otimes \mathbb{I}_{a}
    \right),\nonumber\\
    &U_{\mathrm{diag}}
    =
    (\operatorname{QFT}_J \otimes \mathbb{I}_a)
    U_{f}
    (\operatorname{QFT}_J^\dagger\otimes\mathbb{I}_a),\nonumber\\
    &\widehat U_{BE}:=
    \left(
    \mathbb{I}_S\otimes P_{\mathrm{out}}^\dagger\otimes \mathbb{I}_{a}
    \right)
    \left(
    \mathcal W_U^{(S,J)}\otimes \mathbb{I}_{a}
    \right)^\dagger
    \left(
    \mathbb{I}_S\otimes \widehat U_{\mathrm{diag}}
    \right)
    \left(
    \mathcal W_U^{(S,J)}\otimes \mathbb{I}_{a}
    \right)
    \left(
    \mathbb{I}_S\otimes P_{\mathrm{in}}\otimes \mathbb{I}_{a}
    \right),\nonumber\\
    &\widehat U_{\mathrm{diag}}
    =
    (\operatorname{QFT}_J \otimes \mathbb{I}_a)
    \widehat U_{f}
    (\operatorname{QFT}_J^\dagger\otimes\mathbb{I}_a).\nonumber\\
\end{align}We now estimate the error between the block encodigs of matrix polynomials obtained from $U_{BE}$ and $\widehat U_{BE}$. Since $\bra{0^{m+3}}U_{BE}\ket{0^{m+3}}=\frac{f_d(U)}{\sqrt{2}\eta}$ by \cref{lem: lemma 9}, we obtain
\begin{align}\label{ineq: last}
    &\norm{
    \bra{0^{m+3}}\widehat U_{BE}\ket{0^{m+3}}
    -
    \frac{f_d(U)}{\sqrt2\eta}}_2\nonumber \\
    &=\norm{\bra{0^{m+3}}(\widehat U_{BE}-U_{BE})\ket{0^{m+3}}}_2
    \nonumber\\
    &\le\norm{(\mathbb{I}_S\otimes\mathbb{I}_J\otimes\bra{0}_a)(\widehat U_{BE}-U_{BE})(\mathbb{I}_S\otimes \mathbb{I}_J\otimes\ket{0}_a)}_2
    \nonumber\\
    &\le
    \norm{(\mathbb{I}_J\otimes \bra{0}_a) (\widehat U_f - U_f) (\mathbb{I}_J\otimes \ket{0}_a)}_2
    \nonumber\\
    &\le C_{12}\sqrt d\log d\,\epsilon_{ma}.
\end{align}The first inequality holds since
the ancilla projection yields contraction. The second inequality holds by the unitary invariance. In the last inequality, we use \eqref{ineq: u-u}. By letting $\epsilon_{ma}=2^{-r}$,
the target error $\epsilon$ is achieved whenever
\begin{align}
    C_{13}\sqrt{d}\log d\,2^{-r}
    \le\epsilon
\end{align}with $r=O\left(
    \log d+\log\frac1\epsilon
    \right)$
bits of precision. We denote by $\widehat f(U)$, the matrix polynomial corresponding to $\sqrt{2}\eta\bra{0^{m+3}}\widehat U_{BE}\ket{0^{m+3}}$. Then, by letting $U$ the diagonal matrix with $z_k$'s and setting $\epsilon\rightarrow \frac{\epsilon}{\sqrt{2}}$, we see from \eqref{ineq: last}
\begin{align}
    \frac{\max_{z_k}\abs{\widehat f(z_k)-f_d(z_k)}}{\eta}\le \epsilon,
\end{align}which proves the statement.

\end{proof}

\bibliographystyle{unsrt}
\bibliography{ref}

\end{document}